\documentclass[reqno,11pt]{amsart}
\usepackage{amscd,amssymb,verbatim}

\numberwithin{equation}{section} \theoremstyle{plain}
\theoremstyle{plain}
\newtheorem{Thm}{Theorem}
\newtheorem{Cor}[subsection]{Corollary}
\newtheorem{Lem}[subsection]{Lemma}
\newtheorem{Prop}[subsection]{Proposition}

\theoremstyle{definition}
\newtheorem{Def}[subsection]{Definition}

\theoremstyle{remark}

\newtheorem{rem}[subsection]{Remark}

\theoremstyle{example}
\newtheorem{ex}[subsection]{Example}

\newenvironment{thm}%
          { \begin{Thm}  }%
          { \end{Thm} }

\newenvironment{lem}%
          { \begin{Lem}    }%
          { \end{Lem} }

          { \begin{Prop}  }%
          { \end{Prop} }

          { \begin{Cor} }%
          { \end{Cor} }

          { \begin{Def} }%
          { \end{Def} }

\newcommand{\vcomp}{C_{c}(V)}
\newcommand{\lw}{{\ell^{2}_{\mu}}(V)}

\newcommand{\delswa}{\Delta_{b,\mu;\theta}}
\newcommand{\hmax}{H_{\max}}
\newcommand{\hmin}{H_{\min}}
\newcommand{\Dom}{\operatorname{Dom}}

\title{Self-adjoint extensions of discrete magnetic Schr\"odinger operators}
\author{Ognjen Milatovic, Fran{\c c}oise Truc}
\address{Department of Mathematics
and Statistics \\ University of North Florida \\ Jacksonville, FL
32224 \\ USA.}
\email{omilatov@unf.edu}

\address{Grenoble University\\ Institut Fourier\\
Unit{\'e} mixte
 de recherche CNRS-UJF 5582\\
 BP 74, 38402-Saint Martin d'H\`eres Cedex, France.}
\email{francoise.truc@ujf-grenoble.fr}
\begin{document}
\maketitle

\begin{abstract} Using the concept of intrinsic metric on a locally finite weighted graph, we give sufficient conditions for the magnetic Schr\"odinger operator to be essentially self-adjoint. The present paper is an extension of some recent results proven in the context of graphs of bounded degree.
\end{abstract}
\section{Introduction and the main results}\label{S:main}
\subsection{The setting}\label{SS:setting}
Let $V$ be a countably infinite set. We assume that $V$ is equipped with a measure $\mu\colon V\to (0,\infty)$.
Let $b\colon V\times V\to[0,\infty)$ be a function such that

\medskip

\noindent (i) $b (x, y) = b (y, x)$, for all $x,\,y\in V$;

\medskip

\noindent (ii) $b(x,x)=0$, for all $x\in V$;

\medskip

\noindent (iii)  $\textrm{deg}(x):=\displaystyle\sharp\,\{y\in V\colon b(x,y)>0\}<\infty$, for all $x\in V$. Here, $\sharp\, S$ denotes the number of elements in the set $S$.

\medskip

Vertices $x,\, y\in V$ with $b(x, y) > 0$ are called \emph{neighbors}, and we denote this relationship by $x\sim y$. We call the triple $(V,b,\mu)$ a \emph{locally finite weighted graph}. We assume that $(V,b,\mu)$ is connected, that is, for any $x,\,y\in V$ there exists a path $\gamma$ joining $x$ and $y$. Here, $\gamma$ is a sequence $x_0,\,x_2,\,\dots,x_n\in V$ such that $x=x_0$, $y=x_n$, and $x_{j}\sim x_{j+1}$ for all $0\leq j\leq n-1$.




\subsection{Intrinsic metric}\label{SS:intrinsic-metric} Following~\cite{HKMW} we define a \emph{pseudo metric} to be a map $d\colon V\times V\to [0,\infty)$  such that $d(x,y)=d(y,x)$, for all $x,\,y\in V$; $d(x,x)=0$, for all $x\in V$; and $d(x,y)$ satisfies the triangle inequality. A pseudo-metric $d=d_{\sigma}$ is called a path pseudo-metric if there exists a map $\sigma\colon V\times V\to [0,\infty)$ such that $\sigma(x,y)=\sigma(y,x)$, for all $x,\,y\in V$; $\sigma(x,y)>0$ if and only if $x\sim y$; and
\[
d_{\sigma}=\inf\{l_{\sigma}(\gamma)\colon \gamma=(x_0,x_1,\dots,x_n), n\geq 1, \textrm{is a path connecting }x\textrm{ and }y\},
\]
where the length $l_{\sigma}$ of the path $\gamma=(x_0,x_1,\dots,x_n)$ is given by
\begin{equation}\label{E:l-sigma-def}
l_{\sigma}(\gamma)=\sum_{i=0}^{n-1}\sigma(x_i,x_{i+1}).
\end{equation}

As in~\cite{HKMW} we make the following definitions.

\begin{Def}\label{D:intrinsic} \noindent (i) A pseudo metric $d$ on $(V,b,\mu)$ is called \emph{intrinsic} if
\[
\frac{1}{\mu(x)}\sum_{y\in V}b(x,y)(d(x,y))^2\leq 1,\qquad\textrm{for all }x\in V.
\]
\noindent (ii) An intrinsic path pseudo metric $d=d_{\sigma}$ on $(V,b,\mu)$ is called \emph{strongly intrinsic} if
\[
\frac{1}{\mu(x)}\sum_{y\in V}b(x,y)(\sigma(x,y))^2\leq 1,\qquad\textrm{for all }x\in V.
\]
\end{Def}
\begin{rem} On a locally finite graph $(V,b,\mu)$, the formula
\begin{equation}\label{E:ex-intrinsic}
\sigma_{1}(x,y)=b(x,y)^{-1/2}\min\left\{\frac{\mu(x)}{\textrm{deg}(x)},\frac{\mu(y)}{\textrm{deg}(y)}\right\}^{1/2},\qquad\textrm{with }x\sim y,
\end{equation}
where $\textrm{deg}(x)$ is as in property (iii) of $b(x,y)$, defines a strongly intrinsic path metric; see~\cite[Example 2.1]{HKMW}.
\end{rem}

\subsection{Cauchy boundary} For a path metric $d=d_{\sigma}$ on $V$, we denote the metric completion by $(\widehat{V},\widehat{d})$. As in~\cite{HKMW} we define the \emph{Cauchy boundary} $\partial_{C}V$ as follows: $\partial_{C}V:=\widehat{V}\backslash V$. Note that $(V,d)$ is metrically complete if and only if $\partial_{C}V$ is empty. For a path metric $d=d_{\sigma}$ on $V$ and $x\in V$, we define
\begin{equation}\label{E:dist-boundary}
D(x):=\inf_{z\in \partial_{C}V}d_{\sigma}(x,z).
\end{equation}

\subsection{Inner product} In what follows, $C(V)$ is the set of complex-valued functions on $V$, and $\vcomp$ is the set of finitely supported elements of $C(V)$.
By $\lw$  we denote the space of functions $f\in C(V)$ such that
\begin{equation}\label{E:l-p-def}
\|f\|^{2}:=\sum_{x\in V}\mu(x)|f(x)|^2<\infty,
\end{equation}
where $|\cdot|$ denotes the modulus of a complex number.

In particular, the space $\lw$ is a Hilbert space with the inner product
\begin{equation}\label{E:inner-w}
(f,g):=\sum_{x\in V}\mu(x)f(x)\overline{g(x)}.
\end{equation}


\subsection{Laplacian operator}\label{SS:laplacian-def}
We define the formal Laplacian $\Delta_{b,\mu}\colon C(V)\to C(V)$ on $(V,b, \mu)$ by the formula
\begin{equation}\label{E:ord-lap}
    (\Delta_{b,\mu} u)(x)=\frac{1}{\mu(x)}\sum_{y\in V}b(x,y)(u(x)-u(y)).
\end{equation}
\subsection{Magnetic Schr\"odinger operator}\label{SS:magnetic-schro}
We fix a phase function $\theta \colon V\times V\to [-\pi,\pi]$ such that $\theta(x,y)=-\theta(y,x)$ for all $x\,,y\in V$, and denote $\theta_{x,y}:=\theta(x,y)$.
We define the formal magnetic Laplacian $\delswa\colon C(V)\to C(V)$ on $(V,b,\mu)$ by the formula
\begin{equation}\label{E:magnetic-lap}
    (\delswa u)(x)=\frac{1}{\mu(x)}\sum_{y\in V}b(x,y)(u(x)-e^{i\theta_{x,y}}u(y)).
\end{equation}

We define the formal magnetic Schr\"odinger operator $H\colon C(V)\to C(V)$ by the formula
\begin{equation}\label{E:magnetic-schro}
Hu:= \delswa u +Wu,
\end{equation}
where $W\colon V\to \mathbb{R}$.

\subsection{Statements of the results}
We are ready to state our first result.
\begin{thm}\label{T:main-1}  Assume that $(V, b, \mu)$ is a locally finite, weighted, and connected graph. Let $d=d_{\sigma}$ be an intrinsic path metric on $V$ such that $(V,d)$ is not metrically complete. Assume that there exists a constant $C$ such that
\begin{equation}\label{E:potential-minorant}
W(x) \geq\frac{1}{2(D(x))^2}-C,\quad \textrm{ for all }x\in V,
\end{equation}
where $D(x)$ is as in~(\ref{E:dist-boundary}).
Then $H$ is essentially self-adjoint on $C_c(V)$.
\end{thm}
\begin{rem}
It is possible to find $\mu$, $b$, and a potential $W$ satisfying
$W(x) \geq \frac{k}{2(D(x))^2}~$ with $0<k<1$,
such that $H=\Delta_{b,\mu}+W$ is not essentially self-adjoint; see~\cite[Section 5.3.2]{vtt-11-2}.
\end{rem}

If the graph $(V,b,\mu)$ has a special type of covering, the condition~(\ref{E:potential-minorant}) on $W$ can be relaxed with the help of ``effective potential," as seen in the next theorem. First, we give a description of this special type of covering. In what follows, for a graph $(V,b,\mu)$, we define the set of unoriented edges as $E:=\{\{x,y\}\colon x,\,y\in V\textrm{ and }b(x,y)>0\}$. Sometimes, when we want to emphasize the set $E$, instead of $G=(V,b,\mu)$ we will use the notation $G=(V,E)$.
\begin{Def}\label{D:def-good-cov}
Let $m\in \mathbb{N}$. A \emph{good covering of degree $m$} of $G=(V,E)$
is a family $G_l =(V_l, E_l )_{l \in L}$
of finite connected sub-graphs of
$G$ so that
\begin{enumerate}
  \item[(i)]$ V=\cup_{l\in L} V_l $;
  \item[(ii)] for any $\{x, y\} \in E,$
\[ 0< \# \{ l \in L~|~ \{ x, y \} \in E_l  \} \leq m .\]
\end{enumerate}
\end{Def}
\begin{rem} It is known that a graph with bounded vertex degree admits a good covering; see~\cite[Proposition 2.2]{vtt-11-3}. The graph in Example~\ref{E:ex-2-1} below does not have a bounded vertex degree. Note that this graph has a good covering of degree $m=2$.
\end{rem}

Assume that $(V,b,\mu)$ has a good covering $(V_l, E_l)_{l\in L}$. Let $\theta_l$ be the restriction of $\theta$ to $V_l\times V_l$.  Let $\Delta^{(l)}_{1,\mu;\theta}$ be as in~(\ref{E:magnetic-lap}) with $V=V_l$, $\theta=\theta_{l}$, and $b\equiv 1$. Then $\Delta^{(l)}_{1,\mu;\theta}$ is a bounded and non-negative self-adjoint operator in $\ell^{2}_{\mu}(V_l)$. Let $p_l$ denote the lowest eigenvalue of $\Delta^{(l)}_{1,\mu;\theta}$.
With these notations, for a graph $(V,b,\mu)$ and the phase function $\theta$, we define the \emph{effective potential} corresponding to a good covering $(V_l, E_l)_{l\in L}$ of degree $m$ as follows:
\begin{equation}\label{E:QW}
 W_{e}(x):= \frac{1}{m}  \sum _{\{ l\in L\, |\, x\in V_l \}} p_l
\inf _{\{y,z\} \in E_l}b(y,z).
\end{equation}
We now state our second result.
\begin{thm}\label{T:main-2}  Assume that $(V, b, \mu)$ is a locally finite, weighted, and connected graph. Assume that $(V, b, \mu)$ has a good covering $(V_l, E_l)_{l\in L}$. Let $d=d_{\sigma}$ be an intrinsic path metric on $V$ such that $(V,d)$ is not metrically complete. Assume that there exists a constant $C$ such that
\begin{equation}\label{E:potential-minorant-effective}
W_{e}(x)+W(x) \geq\frac{1}{2(D(x))^2}-C,\quad \textrm{ for all }x\in V,
\end{equation}
where $W_{e}$ is as in~(\ref{E:QW}) and $D(x)$ is as in~(\ref{E:dist-boundary}).
Then $H$ is essentially self-adjoint on $C_c(V)$.
\end{thm}

In the setting of metrically complete graphs, we have the following result:

\begin{thm}\label{T:main-3} Assume that $(V, b, \mu)$ be a locally finite, weighted, and connected graph.  Let $d_{\sigma}$ be a strongly intrinsic path metric on $V$. Let $q\colon V\to [1,\infty)$ be a function satisfying
\begin{equation}\label{E:q-1-2-lipschitz}
|q^{-1/2}(x)-q^{-1/2}(y)|\leq K{\sigma}(x,y),\qquad\textrm{for all }x,\,y\in V\textrm{ such that }x\sim y,
\end{equation}
where $K$ is a constant. Let $H$ be as in~(\ref{E:magnetic-schro}) with $W\colon V\to \mathbb{R}$ satisfying
\begin{equation}\label{E:W-minorant-1}
W(x)\geq -q(x),\quad \textrm{for all }x\in V.
\end{equation}
Let
\begin{equation}\label{E:def-sigma-q}
\sigma_{q}(x,y)=\min\{q^{-1/2}(x),q^{-1/2}(y)\}\cdot\sigma(x,y)
\end{equation}
and let $d_{\sigma_{q}}$ be the path metric corresponding to $\sigma_{q}$. Assume that $(V,d_{\sigma_{q}})$ is metrically complete.
Then $H$ is essentially self-adjoint on $\vcomp$.
\end{thm}

\subsection{Some comments on the existing literature} The notion of intrinsic metric allows us to remove the bounded vertex degree assumption present in~\cite{vtt-11-2,vtt-11-3,Milatovic-12}. More specifically, Theorem~\ref{T:main-1} extends~\cite[Theorem 4.2]{vtt-11-2}, which was proven in the context of graphs of bounded vertex degree for the operator $\Delta_{b,\mu}+W$, with $\Delta_{b,\mu}$ as in~(\ref{E:ord-lap}). Theorem~\ref{T:main-2} is an extension of~\cite[Theorem 3.1]{vtt-11-3}, which was proven in the context of graphs of bounded vertex degree for the operator $\delswa$. In this regard, the first two results of the present paper answer a question posed in~\cite[Section 5]{vtt-11-3}. Theorem~\ref{T:main-3} extends~\cite[Theorem 1]{Milatovic-12}, which was proven in the context of graphs of bounded vertex degree for the operator $\delswa+W$ with $W$ as in~(\ref{E:W-minorant-1}). We should also mention that in the context of locally finite graphs (with an assumption on $b$ and $\mu$ originating from~\cite{Masamune-09}), a sufficient condition for the essential self-adjointness of a semi-bounded from below operator $\delswa+W$ is given in ~\cite[Theorem 1.2]{Milatovic-11}. Another sufficient condition for the essential self-adjointness of $\delswa+W$  is given in~\cite[Proposition 2.2]{Golenia-11}:\emph{ Let $(V,b,\mu)$ be a locally finite weighted graph. Let $W\colon V\to\mathbb{R}$ and $\delta>0$. Take $\lambda\in\mathbb{R}$ so that
\begin{equation}\label{E:g-condition-1}
\{x\in V\colon \lambda+ \operatorname{Deg}(x)+W(x)=0\}=\emptyset,
\end{equation}
where $\operatorname{Deg}(x)$ denotes the ``weighted degree"
\begin{equation}\label{E:Deg-def}
\operatorname{Deg}(x):=\frac{1}{\mu(x)}\sum_{y\in V}b(x,y),\qquad x\in V.
\end{equation}
Suppose that for every sequence of vertices $\{y_1,\,y_2,\dots\}$ such that $y_j\sim y_{j+1}$, $j\geq 1$, the following property holds:
\begin{equation}\label{E:g-condition-2}
\sum_{n=1}^{\infty}((a_n)^2\mu(y_n))=\infty, \quad\textrm{ where} \quad a_n:=\prod_{j=1}^{n-1}\left(\frac{\delta}{\operatorname{Deg}(y_j)}+\left|1+\frac{\lambda+W(y_j)}{\operatorname{Deg}(y_j)}\right|\right), \quad n\geq 2,
\end{equation}
and $a_1:=1$. Then $\delswa+W$ is essentially self-adjoint on $\vcomp$.  }

Note that~\cite[Proposition 2.2]{Golenia-11} allows potentials that are unbounded from below. We mention that Example~\ref{E:ex-2-1} below describes a situation where Theorem~\ref{T:main-2} is applicable, while neither~\cite[Theorem 1.2]{Milatovic-11} nor~\cite[Proposition 2.2]{Golenia-11} is applicable. Additionally, Example~\ref{E:ex-3} below describes a situation where Theorem~\ref{T:main-3} is applicable, while neither~\cite[Theorem 1.2]{Milatovic-11} nor~\cite[Proposition 2.2]{Golenia-11} is applicable.

The recent study~\cite{HKMW} is concerned with the operator $\Delta_{b,\mu}$ as in~(\ref{E:ord-lap}), with  property (iii) of $b$ (see Section~\ref{SS:setting} above) replaced by the following more general condition:
\[
\sum_{y\in V}b(x,y)<\infty,\qquad\textrm{for all }x\in V.
\]
Using the notion of intrinsic distance $d$ with finite jump size, the authors of~\cite{HKMW} show that if the weighted degree~(\ref{E:Deg-def})
is bounded on balls defined with respect to any such distance $d$, then $\Delta_{b,\mu}$ is essentially self-adjoint. In the context of a locally finite graph, the authors of~\cite{HKMW} show that if the graph is metrically complete in any intrinsic path metric with finite jump size, then $\Delta_{b,\mu}$ is essentially self-adjoint. In the metrically incomplete case, one of the results of~\cite{HKMW} shows that if the Cauchy boundary has finite capacity, then $\Delta_{b,\mu}$ has a unique Markovian extension if and only if the Cauchy boundary is polar (here, ``Cauchy boundary is polar" means that the Cauchy boundary has zero capacity). Another result of~\cite{HKMW} shows that if the upper Minkowski codimension of the Cauchy boundary is greater than 2, then the Cauchy boundary is polar. Additionally, we should mention that the authors of~\cite{HKMW} prove Hopf--Rinow-type theorem  for locally finite weighted graphs with a path pseudo metric.

In recent years, various authors have developed independently the concept of intrinsic metric on a graph. The definition given in the present paper can be traced back to the work~\cite{flw}. For applications of intrinsic metrics in various contexts, see, for instance,~\cite{BKW-12,folz-11,folz-12-1,folz,GHM-11,HKLW-preprint-12,Huang-12,Huang-12-preprint,MU-11}.

With regard to the problem of self-adjoint extensions of adjacency, (magnetic) Laplacian and Schr\"odinger-type operators on infinite graphs, we should mention that there has been a lot of interest in this area in the past few years. For references to the literature on this topic, see, for instance,~\cite{vtt-11-2,vtt-11-3,Golenia-11,HKLW-preprint-11,HKMW,Masamune-09,Milatovic-12,Torki-10}.

\section{Proof of Theorem~\ref{T:main-1}}
In this section, we modify the proof of~\cite[Theorem 4.2]{vtt-11-2}. Throughout the section, we assume that the hypotheses of Theorem~\ref{T:main-1}
are satisfied. We begin with the following lemma, whose proof is given in~\cite[Lemma 3.3]{vtt-11-3}.

\begin{lem}\label{L:nen}
Let $H$ be as in~(\ref{E:magnetic-schro}), let $v\in\lw$ be a weak solution of $Hv=0$, and let $f\in C_{c}(V)$ be a real-valued function.
Then the following equality holds:
\begin{equation}\label{E:ute}
(fv, \,H (fv)) = \frac{1}{2}\sum_{x\in V}\sum_{y\sim x } b(x,y)\,\textrm{\emph{Re }}[e^{-i\theta(x,y)}v(x) \overline{ v(y)}](f(x)-f(y))^2.
\end{equation}
\end{lem}

The key ingredient in the proof of Theorem~\ref{T:main-1} is the Agmon-type estimate given in the next lemma, whose proof, inspired by an idea of~\cite{Nen}, is based on the technique developed in~\cite{Col-Tr} for magnetic Laplacians on an open
set with compact boundary in $\mathbb{R}^{n}$.

\begin{lem}\label{L:Hor}
Let $\lambda\in\mathbb{R}$ and let $v\in\lw$ be a weak solution of $(H-\lambda)v=0$.
Assume that that there exists a constant $c_1>0$ such that, for all $u \in C_c(V)$,
\begin{equation}\label{E:bou}
(u,\, (H-\lambda) u )  \geq
 \dfrac{1}{2}\sum_{x\in V}\max \left(\dfrac{1}{D(x)^2},1
\right) \mu(x) |u(x)|^2  +  c_1 \|u\|^2.
\end{equation}
Then $v\equiv 0$.
\end{lem}

\begin{proof}
Let $\rho$ and $R$ be numbers satisfying
$0< \rho < 1/2$ and  $ 1 < R < +\infty$.
For any $\epsilon >0$, we define the  function $f_{\epsilon}\colon V \rightarrow \mathbb{R}$
by $f_{\epsilon}(x)=F_{\epsilon}(D(x))$, where $D(x)$ is as in~(\ref{E:dist-boundary}) and $F_{\epsilon}\colon{\mathbb{R}}^{+} \rightarrow \mathbb{R}$ is the
continuous piecewise affine
function defined by
\[
F_{\epsilon}(s)= \left\{
\begin{array}{l}
0  {\rm ~ for~}  s\leq \epsilon   \\
\rho (s-\epsilon)/(\rho  - \epsilon  )   {\rm ~ for~} \epsilon \leq s \leq  \rho  \\
s    {\rm ~ for~} \rho  \leq s \leq  1  \\
1   {\rm ~ for~} 1 \leq s \leq R  \\
 R+1 -s   {\rm ~ for~} R \leq s \leq R+1 \\
0   {\rm ~ for~ }  s \geq R+1
\end{array}
\right.
\]

We first note that by the definition of $F_{\epsilon}$ and continuity of $D(x)$, the support of $f_{\epsilon}$ is compact. Now by~\cite[Lemma A.3(b)]{HKMW} it follows that the support of $f_{\epsilon}$ finite. Using Lemma~\ref{L:nen} with $H-\lambda$ in place of $H$, the inequality
$$ \textrm{Re }[e^{-i\theta(x,y)}v(x)\overline{v(y)}] \leq \frac{1}{2}(|v(x)|^2+|v(y)|^2),$$
and Definition~\ref{D:intrinsic}(i) we have
\begin{eqnarray}\label{E:equ-rhs}
&&(f_{\epsilon} v,\, (H-\lambda) (f_{\epsilon} v)) \leq
\frac{1}{2} \sum _{x\in V } \sum_{y\sim x} b(x,y)|v(x)|^2(f_{\epsilon}(x)-f_{\epsilon}(y))^2\nonumber\\
&&\leq \frac{\rho^2}{2(\rho-\epsilon)^2}\sum _{x\in V} \sum_{y\sim x} |v(x)|^2 b(x,y)(d(x,y))^2\leq \frac{\rho^2}{2(\rho-\epsilon)^2}\sum _{x\in V}\mu(x) |v(x)|^2,
\end{eqnarray}
where the second inequality uses the fact that
$f_{\epsilon}$ is a $\beta$-Lipschitz function with $\beta={\rho}/({\rho-\epsilon})$.

On the other hand, using the definition of $f_{\epsilon}$ and the assumption~(\ref{E:bou}) we have
\begin{equation} \label{E:equ-lhs}
(f_{\epsilon}v,\,(H-\lambda) (f_{\epsilon} v))
\geq \frac{1}{2} \sum _{\rho\leq D(x) \leq R }\mu(x)|v(x)|^2 +c_1 \| f_{\epsilon} v
\|^2.
\end{equation}
We now combine (\ref{E:equ-lhs}) and (\ref{E:equ-rhs}) to get
\[
\frac{1}{2} \sum _{\rho\leq D(x) \leq R }\mu(x)|v(x)|^2 +c_1 \| f_{\epsilon} v\|^2\leq \frac{\rho^2}{2(\rho-\epsilon)^2}\sum _{x\in V}\mu(x) |v(x)|^2.
\]
We fix $\rho$ and $R$, and let $\epsilon \to 0+$. After that, we let  $\rho \to 0+ $ and $R \to +\infty$.
As a result, we get $v\equiv 0$.
\end{proof}

\noindent\textbf{Conclusion of the proof of Theorem~\ref{T:main-1}.} Since $\delswa|_{\vcomp}$ is a non-negative operator, for all $u \in C_{c}(V)$, we have
\[
(u,\, Hu) \geq \sum_{x\in V} \mu(x) W(x) |u(x)|^2,
\]
and, hence, by assumption~(\ref{E:potential-minorant}) we get:
\begin{eqnarray}\label{E:bou-new}
&&(u,\, (H-\lambda) u)\geq \dfrac{1}{2}\sum_{x\in V} \frac{1}{D(x)^2}\mu(x)|u(x)|^2-(\lambda+C)\|u\|^2\nonumber\\
&&\geq  \dfrac{1}{2}\sum_{x\in V} \max\left(\frac{1}{D(x)^2},1\right)\mu(x)|u(x)|^2 -(\lambda+C+1/2)\|u\|^2.
\end{eqnarray}
Choosing, for instance, $\lambda=-C-3/2$ in~(\ref{E:bou-new}) we get the
inequality (\ref{E:bou}) with $c_1=1$.

Thus, $(H-\lambda)|_{\vcomp}$ with $\lambda=-C-3/2$ is a symmetric operator satisfying
$(u,\, (H-\lambda) u)\geq \|u\|^2$, for all $u\in\vcomp$. In this case, it is known (see~\cite[Theorem X.26]{rs}) that the essential self-adjointness of $(H-\lambda)|_{\vcomp}$ is equivalent to the following statement: if $v\in\lw$ satisfies $(H-\lambda)v=0$, then $v=0$. Thus, by Lemma~\ref{L:Hor}, the operator $(H-\lambda)|_{\vcomp}$ is essentially self-adjoint. Hence, $H|_{\vcomp}$ is essentially self-adjoint. $\hfill\square$

\section{Proof of Theorem~\ref{T:main-2}}
Throughout the section, we assume that the hypotheses of Theorem~\ref{T:main-2} are satisfied. We begin with the following lemma.

\begin{lem}\label{L:effective-potential} Let $(V_l,E_l)_{l\in L}$ be a good covering of degree $m$ of $(V, b, \mu)$, let $H$ be as in~(\ref{E:magnetic-schro}), and let $W_{e}$ be as in~(\ref{E:QW}).
Then, for all $u\in\vcomp$ we have
\begin{equation}\label{E:to-prove}
(u, Hu) \geq \sum_{x\in V}\mu(x) (W_{e}(x) +W(x)) |u(x)|^2.
\end{equation}
\end{lem}

\begin{proof}
It is well known that
\[
(u,Hu)=\sum_{\{x,y\}\in E}b(x,y)|u(x)-e^{i\theta(x,y)}u(y)|^2+\sum_{x\in V}\mu(x)W(x)|u(x)|^2,
\]
where $E$ is the set of unoriented edges of $(V, b, \mu)$. Thus, using the definition of the covering $(V_l,E_l)_{l\in L}$ of degree $m$ and the definition of $p_l$ we have
\begin{eqnarray*}
&&(u, Hu)\geq \frac{1}{m} \sum _{ l \in L }
\sum _{\{x,y\} \in E_l }b(x,y) |u(x)-e^{i\theta(x,y)}u(y)|^2
+\sum_{x\in V}\mu(x)W(x)|u(x)|^2\nonumber\\
&&\geq \frac{1}{m}\sum _{l\in L}\left( \left( \inf_{\{y,z\} \in E_l}
b(y,z)\right)p_l\sum _{x\in V_l } \mu(x) |u(x)|^2\right)+\sum_{x\in V}\mu(x)W(x)|u(x)|^2,\\
\end{eqnarray*}
which together with~(\ref{E:QW}) gives~(\ref{E:to-prove}).
\end{proof}

\noindent\textbf{Conclusion of the proof of Theorem~\ref{T:main-2}.}
By Lemma~\ref{L:effective-potential} and assumption~(\ref{E:potential-minorant-effective}), for all $u \in \vcomp$ we have
\begin{eqnarray*}
&&(u,\, (H-\lambda)u) \geq \sum_{x\in V} \mu(x) (W_{e}(x)+W(x)-\lambda) |u(x)|^2\\
&&\geq\frac{1}{2}\sum_{x\in V} \max\left(\frac{1}{D(x)^2},1 \right)
\mu(x) |u(x)|^2-(C+\lambda+1/2)\|u\|^2.
\end{eqnarray*}
From hereon we proceed in the same way as in the the proof of Theorem~\ref{T:main-1}. $\hfill\square$

\section{Proof of Theorem~\ref{T:main-3}}
In this section we modify the proof of~\cite[Theorem 1]{Milatovic-12}, which is based on the technique of~\cite{sh2} in the context of Riemannian manifolds. Throughout the section, we assume that the hypotheses of Theorem~\ref{T:main-3} are satisfied.

We begin with the definitions of minimal and maximal operators associated with the expression~(\ref{E:magnetic-schro}). We define the operator $\hmin$ by the formula $\hmin u:=Hu$, for all $u\in\Dom(\hmin):=\vcomp$.  As $W$ is real-valued, it follows easily that the operator $\hmin$ is symmetric in $\lw$. We define $\hmax:=(\hmin)^{*}$, where $T^*$ denotes the adjoint of operator $T$. Additionally, we define
$\mathcal{D}:=\{u\in \lw\colon Hu\in \lw\}$. Then, the following hold: $\Dom(\hmax)=\mathcal{D}$ and $\hmax u=Hu$ for all $u\in\mathcal{D}$; see, for instance,~\cite[Section 3]{Milatovic-12} or~\cite[Section 3]{Torki-10} for details. Furthermore, by~\cite[Problem V.3.10]{Kato80}, the operator  $\hmin$ is essentially self-adjoint if and only if
\begin{equation}\label{E:h-max-symmetric}
    (\hmax u,v)=(u,\hmax v),\qquad\textrm{for all }u\,,v\in\Dom(\hmax).
\end{equation}

In the setting of graphs of bounded vertex degree, the following proposition was proven in~\cite[Proposition 12]{Milatovic-12}.
\begin{Prop}\label{P:info-dom-hmax} If $u\in\Dom(\hmax)$, then
\begin{align}\label{E:info-dom-hmax}
\sum_{x,y\in V}b(x,y)\min\{q^{-1}(x),q^{-1}(y)\}|u(x)-e^{i\theta_{x,y}}u(y)|^2\leq 4(\|Hu\|\|u\|+(K^2+1)\|u\|^2),
\end{align}
where $H$ is as in~(\ref{E:magnetic-schro}) and $K$ is as in~(\ref{E:q-1-2-lipschitz}).
\end{Prop}

Before proving Proposition~\ref{P:info-dom-hmax}, we define a sequence  of cut-off functions.   Let  $d_{\sigma}$  and $d_{\sigma_{q}}$ be as in the hypothesis of Theorem~\ref{T:main-3}. Fix $x_0\in V$ and  define
\begin{equation}\label{E:cut-off}
\chi_n(x):=\left(\left(\frac{2n-d_{\sigma}(x_0,x)}{n}\right)\vee 0\right)\wedge 1,\qquad x\in V,\quad n\in \mathbb{Z}_{+}.
\end{equation}

Denote
\begin{equation}\label{E:nbd-x-0}
B^{\sigma}_{n}(x_0):=\{x\in V\colon d_{\sigma}(x_0,x)\leq n\}.
\end{equation}
The sequence $\{\chi_n\}_{n\in\mathbb{Z}_{+}}$ satisfies the following properties: (i) $0\leq \chi_n(x)\leq 1$, for all $x\in V$; (ii)  $\chi_n(x)=1$ for $x\in B^{\sigma}_{n}(x_0)$ and  $\chi_n(x)=0$ for $x\notin B^{\sigma}_{2n}(x_0)$; (iii) for all $x\in V$, we have $\displaystyle\lim_{n\to\infty}\chi_n(x)=1$; (iv) the functions $\chi_n$ have finite support; and (v) the functions $\chi_n$ satisfy the inequality
\begin{equation}\nonumber
|\chi_n(x)-\chi_n(y)|\leq \frac{\sigma(x,y)}{n},\qquad\textrm{for all }x\sim y.
\end{equation}
 The properties (i)--(iii) and (v) can be checked easily. By hypothesis, we know that $(V,d_{\sigma_{q}})$ is a complete metric space and, thus, balls with respect to $d_{\sigma_{q}}$ are finite; see, for instance,~\cite[Theorem A.1]{HKMW}. Let $B^{\sigma_{q}}_{2n}(x_0)$ be as in~(\ref{E:nbd-x-0}) with $d_{\sigma}$ replaced by $d_{\sigma_{q}}$. Since $q\geq 1$ it follows that $B^{\sigma}_{2n}(x_0)\subseteq B^{\sigma_{q}}_{2n}(x_0)$. Thus, property (iv) is a consequence of property (ii) and the finiteness of $B^{\sigma_{q}}_{2n}(x_0)$.

 \medskip

\noindent\textbf{Proof of Proposition~\ref{P:info-dom-hmax}.} Let $u\in\Dom(\hmax)$ and let $\phi\in\vcomp$ be a real-valued function. Define
\begin{equation}\label{E:sum-I}
I:=\left(\sum_{x,y\in V}b(x,y)|u(x)-e^{i\theta_{x,y}}u(y)|^2((\phi(x))^2+(\phi(y))^2)\right)^{1/2}.
\end{equation}

We will first show that
\begin{align}\label{E:sum-I-show}
&I^2\leq 4|(\phi^2Hu,u)|+4(\phi^2qu,u)\nonumber\\
&+\sqrt{2}I\left(\sum_{x,y\in V}b(x,y)(\phi(x)-\phi(y))^2|(u(x)+e^{i\theta_{x,y}}u(y)|^2\right)^{1/2}.
\end{align}

To do this, we first note that
\begin{align}\label{E:I-squared-equality}
I^2&=4(\phi^2Hu,u)-4(\phi^2Wu,u)\nonumber\\
&+\sum_{x,y\in V}b(x,y)(e^{i\theta_{x,y}}u(y)-u(x))(e^{-i\theta_{x,y}}\overline{u(y)}+\overline{u(x)})((\phi(x))^2-(\phi(y))^2),
\end{align}
which can be checked by expanding the terms under summations on both sides of the equality and using the properties $b(x,y)=b(y,x)$ and $\theta(x,y)=-\theta(y,x)$. The details of this computation can be found in the proof of~\cite[Proposition 12]{Milatovic-12}.

The inequality~(\ref{E:sum-I-show}) is obtained from~(\ref{E:I-squared-equality}) by using~(\ref{E:W-minorant-1}), the factorization
\[
(\phi(x))^2-(\phi(y))^2=(\phi(x)-\phi(y))(\phi(x)+\phi(y)),
\]
Cauchy--Schwarz inequality, and
\[
(\phi(x)+\phi(y))^2\leq 2(\phi^2(x)+\phi^2(y)).
\]

Let $\chi_{n}$ be as in~(\ref{E:cut-off}) and let $q$ be as in~(\ref{E:W-minorant-1}). Define
\begin{equation}\label{E:cut-off-phi-n}
\phi_{n}(x):=\chi_{n}(x)q^{-1/2}(x).
\end{equation}
By property (iv) of $\chi_n$ it follows that $\phi_n$ has finite support. By property (i) of $\chi_n$  and since $q\geq 1$, we have
\begin{equation}\label{E:cut-off-phi-n-property}
0\leq \phi_n(x)\leq q^{-1/2}(x)\leq 1,\qquad\textrm{for all }x\in V.
\end{equation}

By property (iii) of $\chi_n$  we have

\begin{equation}\label{E:cut-off-phi-n-limit}
\displaystyle\lim_{n\to\infty}\phi_n(x)=q^{-1/2}(x),\qquad\textrm{for all }x\in V.
\end{equation}

By~(\ref{E:q-1-2-lipschitz}), properties (i) and (v) of $\chi_n$, and the inequality $q\geq 1$, we have
\begin{align}\label{E:d-phi-n}
|\phi_{n}(x)-\phi_n(y)|\leq\left(\frac{1}{n}+K\right)\sigma(x,y),\qquad \textrm{for all }x\sim y,
\end{align}
where $K$ is as in~(\ref{E:q-1-2-lipschitz}). We will also use the inequality
\begin{align}\label{E:u-sharp-sigma-bound}
|e^{i\theta_{x,y}}u(y)+{u(x)}|^2\leq 2(|u(x)|^2+|u(y)|^2).
\end{align}

By~(\ref{E:d-phi-n}),~(\ref{E:u-sharp-sigma-bound}), and Definition~\ref{D:intrinsic}(ii), we get
\begin{align}\label{E:u-sharp-sigma-bound-1}
&\left(\sum_{x,y\in V}b(x,y)(\phi_n(x)-\phi_n(y))^2|(u(x)+e^{i\theta_{x,y}}u(y)|^2\right)^{1/2}\nonumber\\
&\leq {\sqrt{2}}\left(\frac{1}{n}+K\right)\left(\sum_{x,y\in V}b(x,y)(\sigma(x,y))^2(|u(x)|^2+|u(y)|^2)\right)^{1/2}\nonumber\\
&= 2\left(\frac{1}{n}+K\right)\left(\sum_{x,y\in V}b(x,y)(\sigma(x,y))^2|u(x)|^2\right)^{1/2}\nonumber\\
&\leq 2\left(\frac{1}{n}+K\right)\left(\sum_{x\in V}\mu(x)|u(x)|^2\right)^{1/2}
\end{align}

By~(\ref{E:sum-I-show}) with $\phi=\phi_n$,~(\ref{E:u-sharp-sigma-bound-1}), and~(\ref{E:cut-off-phi-n-property}), we obtain
\begin{equation}\label{E:quadratic-inequality}
I_n^2\leq 4\|Hu\|\|u\|+4\|u\|^2+2\sqrt{2}{I_{n}}\left(\frac{1}{n}+K\right)\|u\|,
\end{equation}
for all $u\in\Dom(\hmax)$, where $I_{n}$ is as in~(\ref{E:sum-I}) with $\phi=\phi_n$.

Using the inequality $ab\leq \frac{a^2}{4}+{b^2}$ with $a=\sqrt{2}I_{n}$ in the third term on the right-hand side of~(\ref{E:quadratic-inequality}) and rearranging, we obtain
\begin{equation}\label{E:quadratic-inequality-2}
I_n^2\leq 8\left(\|Hu\|\|u\|+\left(\left(\frac{1}{n}+K\right)^2+1\right)\|u\|^2\right).
\end{equation}

Letting $n\to\infty$ in~(\ref{E:quadratic-inequality-2}) and using~(\ref{E:cut-off-phi-n-limit}) together with Fatou's lemma, we get

\begin{align}\label{E:quadratic-inequality-3}
&\sum_{x,y\in V}b(x,y)|u(x)-e^{i\theta_{x,y}}u(y)|^2(q^{-1}(x)+q^{-1}(y))\nonumber\\
&\leq 8\left(\|Hu\|\|u\|+({K^2}+1)\|u\|^2\right).
\end{align}
Since
\[
2\min\{q^{-1}(x),q^{-1}(y)\}\leq q^{-1}(x)+q^{-1}(y),\qquad\textrm{for all }x,\,y\in V,
\]
the inequality~(\ref{E:info-dom-hmax}) follows directly from~(\ref{E:quadratic-inequality-3}). $\hfill\square$

\medskip

\noindent\textbf{Continuation of the proof of Theorem~\ref{T:main-3}.} Our final goal is to prove~(\ref{E:h-max-symmetric}). Let $d_{\sigma_{q}}$ be as in the hypothesis of Theorem~\ref{T:main-3}. Fix $x_0\in V$ and define
\begin{equation}\label{E:function-P}
    P(x):=d_{\sigma_{q}}(x_0,x),\qquad x\in V.
\end{equation}

In what follows, for a function $f\colon V\to\mathbb{R}$ we define $f^{+}(x):=\max\{f(x),0\}$. Let $u\,,v\in\Dom(\hmax)$, let $s>0$, and define
\begin{equation}\label{E:sum-J-s}
J_{s}:=\sum_{x\in V}\left(1-\frac{P(x)}{s}\right)^{+}\left((Hu)(x)\overline{v(x)}-u(x)\overline{(Hv)(x)}\right)\mu(x),
\end{equation}
where $P$ is as in~(\ref{E:function-P}) and $H$ is as in~(\ref{E:magnetic-schro}).

Since $(V,d_{\sigma_{q}})$ is a complete metric space, by~\cite[Theorem A.1]{HKMW} it follows that the set
\[
U_{s}:=\{x\in V\colon P(x)\leq s\}
\]
is finite. Thus, for all $s>0$, the summation in~(\ref{E:sum-J-s}) is performed over finitely many vertices.

The following lemma follows easily from the definition of $J_{s}$ and the dominated convergence theorem; see the proof of~\cite[Lemma 13]{Milatovic-12} for details.

\begin{lem}\label{L:limit-J-s} Let $J_{s}$ be as in~(\ref{E:sum-J-s}). Then
\begin{equation}\label{E:limit-J-s-symmetric}
\lim_{s\to+\infty}J_{s}=(Hu,v)-(u,Hv).
\end{equation}
\end{lem}

In what follows, for $u\in\Dom(\hmax)$, define
\begin{equation}\label{E:T-u-def}
T_u:=\left(\sum_{x,y\in V}b(x,y)\min\{q^{-1}(x),q^{-1}(y)\}|u(x)-e^{i\theta_{x,y}}u(y)|^2\right)^{1/2}.
\end{equation}
Note that $T_{u}$ is finite by Proposition~\ref{P:info-dom-hmax}.

\begin{lem}\label{L:J-s-bound} Let $u,\, v\in\Dom(\hmax)$, let $T_{u}$ and $T_v$ be as in~(\ref{E:T-u-def}), and let $J_{s}$ be as in~(\ref{E:sum-J-s}). Then
\begin{align}\label{E:J-s-bound}
|J_s| \leq \frac{1}{2s}(\|v\|T_{u}+\|u\|T_{v}).
\end{align}
\end{lem}
\begin{proof}
A computation shows that
\begin{align*}\label{E:J-s-rewrite-final}
2J_{s}=&\sum_{x,y\in V}\left(\left(1-P(x)/s\right)^{+}-\left(1-P(y)/s\right)^{+}\right)
b(x,y)\left((e^{-i\theta_{x,y}}\overline{v(y)}-\overline{v(x)})u(x)\right.\nonumber\\
&\left.-(e^{i\theta_{x,y}}u(y)-u(x))\overline{v(x)}\right),
\end{align*}
which, together with the triangle inequality and property
\[
|f^{+}(x)-g^{+}(x)|\leq|f(x)-g(x)|,
\]
leads to the following estimate:

\begin{align}\label{E:J-s-pre-final}
2|J_{s}|\leq &\frac{1}{s} \sum_{x,y\in V}b(x,y)|P(x)-P(y)|\left(|e^{i\theta_{x,y}}v(y)-v(x)||u(x)|\right.\nonumber\\
&\left.+|e^{i\theta_{x,y}}u(y)-u(x)||v(x)|\right).
\end{align}

By~(\ref{E:function-P}) and~(\ref{E:def-sigma-q}), for all $x\sim y$ we have
\begin{align}\label{E:distances-P-d}
&|P(x)-P(y)|\leq d_{\sigma_{q}}(x,y)\leq \sigma_{q}(x,y)=\min\{q^{-1/2}(x),q^{-1/2}(y)\}\cdot\sigma(x,y).
\end{align}
To obtain~(\ref{E:J-s-bound}), we combine (\ref{E:J-s-pre-final}) and~(\ref{E:distances-P-d}) and use Cauchy--Schwarz inequality together with Definition~\ref{D:intrinsic}(ii).
\end{proof}

\noindent\textbf{The end of the proof of Theorem~\ref{T:main-3}.} Let $u\in\Dom(\hmax)$ and $v\in\Dom(\hmax)$. By the definition of $\hmax$, it follows that $Hu\in\lw$ and $Hv\in\lw$. Letting $s\to+\infty$ in~(\ref{E:J-s-bound}) and using the finiteness of $T_{u}$ and $T_{v}$, it follows that $J_{s}\to 0$ as $s\to+\infty$. This, together with~(\ref{E:limit-J-s-symmetric}), shows~(\ref{E:h-max-symmetric}). $\hfill\square$

\section{Examples}
In this section we give some examples that illustrate the main results of the paper. In what follows, for $x\in\mathbb{R}$, the notation $\lceil{x}\rceil$ denotes the smallest integer $N$ such that $N\geq x$. Additionally, $\lfloor{x}\rfloor$ denotes the greatest integer $N$ such that $N\leq x$.

\begin{ex}\label{E:ex-2-1}
\emph{In this example we consider the graph $G=(V,E)$ whose vertices $x_{j,k}$ are arranged in a ``triangular" pattern so that the first row contains $x_{1,1}$; for $2\leq j\leq 4$, the $j$-th row contains $x_{j,1}$ and $x_{j,2}$; for $5\leq j\leq 9$, the $j$-th row contains $x_{j,1}$, $x_{j,2}$, and $x_{j,3}$; for $10\leq j\leq 16$, the $j$-th row contains $x_{j,1}$, $x_{j,2}$, $x_{j,3}$, and $x_{j,4}$; and so on. There are two types of edges in the graph: (i) for every $j\geq 1$, we have $x_{j,1}\sim x_{j+1,k}$ for all $1\leq k\leq \lceil{(j+1)^{1/2}}\rceil$; (ii) for every $j\geq 2$, we have the ``horizontal" edges $x_{j,k}\sim x_{j,k+1}$, for all $1\leq k\leq \lceil{j^{1/2}}\rceil-1$. Clearly, $G$ does not have a bounded vertex degree.}

\emph{Let $T=(V_T,E_T)$ be the subgraph of $G$ whose set of edges $E_T$ consists of type-(i) edges of $G$ described above. Note that $T$ is a spanning tree of $G$. Additionally, note that for every type-(ii) edge $e$ of $G$ the following are true: (i) $e\notin E_T$ and (ii) there is a unique $3$-cycle (a cycle with 3 vertices) that contains $e$. Thus, by~\cite[Lemma 2.2]{vtt-11-3}, the corresponding $3$-cycles, which we enumerate by $\{C_l\}_{l\in\mathbb{Z}_{+}}$, form a basis for the space of cycles of $G$. Furthermore, by Definition~\ref{D:def-good-cov}, the family $\{C_l=(V_l, E_l)\}_{l\in\mathbb{Z}_{+}}$ is a good covering of degree $m=2$ of $G$.  Following~\cite[Proposition 2.4(i)]{vtt-11-3} and~\cite[Lemma 2.9]{vtt-11-3}, we define the phase function $\theta \colon V_{l}\times V_{l}\to [-\pi,\pi]$ satisfying the following properties: (i) if an edge $\{x,y\}$ belongs to $E_l\backslash E_T$, we have $\theta(x,y)=-\theta(y,x)$; (ii) if $\{x,y\}\in E_{T}$, we have $\theta(x,y)=0$; and (iii) $p_{l}=|1-e^{i\pi/3}|^2=1$, where $p_l$ is as in~(\ref{E:QW}) with $G_l$ replaced by $C_l$.}

\emph{With this choice of $p_l$ and using the good covering $\{C_l\}_{l\in\mathbb{Z}_{+}}$ of degree $m=2$, the definition of the effective potential~(\ref{E:QW}) simplifies to
\begin{equation}\label{E:QW-simple}
 W_{e}(x):= \frac{1}{2}  \sum _{\{ l\in L\, |\, x\in V_l \}}
\inf _{\{y,z\} \in E_l}b(y,z).
\end{equation}
Let $\{b_j\}_{j\in\mathbb{Z}_{+}}$ be an increasing sequence of positive numbers. We define (i) $b(x,y)=b_j$ if $x\sim y$ and $x$ is in the $j$-th row and $y$ is in the $(j+1)$-st row; (ii) $b(x,y)=b_j$ if $x\sim y$ and $x$ and $y$ are both in the $(j+1)$-st row; (iii) $b(x,y)=0$, otherwise. With this choice of $b(x,y)$, we have $W_{e}(x_{1,1})=b_1/{2}$. Additionally, since $b_j$ is an increasing sequence of positive numbers, using~(\ref{E:QW-simple}) it is easy to see that if a vertex $x$ is in the $j$-th row, then
\begin{equation}\label{E:minorant-eff}
W_{e}(x)\geq \frac{1}{2}b_{j-1},\quad \textrm{for all }j\geq 2.
\end{equation}
Let $0<\beta<3/4$, and set $\mu(x):=j^{-2\beta}$ if the vertex $x$ is in the $j$-th row. Let $\alpha>0$ satisfy $\alpha+2\beta>3/2$, and set $b_j:=j^{\alpha}$, for all $j\in \mathbb{Z}_{+}$.   With this choice of $b(x,y)$ and $\mu(x)$,  let $\sigma_{1}(x,y)$ be as in~(\ref{E:ex-intrinsic}) and let $d_{\sigma_1}$ be the intrinsic path metric associated with $\sigma_1$ as in Section~\ref{SS:intrinsic-metric}. As there are $\lfloor{\sqrt{j}}\rfloor +3$ edges departing from  the vertex $x_{j,1}$, we have
\begin{equation}\label{E:sigma-1}
\sigma_{1}(x_{j,1};\,x_{j+1,1})=j^{-\alpha/2}(j+1)^{-\beta}(\lfloor{\sqrt{j+1}}\rfloor+3)^{-1/2},\qquad\textrm{for all }j\in \mathbb{Z}_{+}.
\end{equation}}
\emph{Additionally, note that the path $\gamma=(x_{1,1};\,x_{2,1};\, x_{3,1};\dots)$ is a geodesic with respect to the path metric $d_{\sigma_1}$, that is,
$d_{\sigma_1}(x_{1,1};\,x_{j,1})=l_{\sigma_{1}}(x_{1,1};\,x_{2,1};\dots;\,x_{j,1})$ for all $j\in\mathbb{Z}_{+}$, where $l_{\sigma_{1}}$ is as in~(\ref{E:l-sigma-def}). Since $\alpha+2\beta>3/2$, it follows that
\[
\sum_{j=1}^{\infty}j^{-\alpha/2}(j+1)^{-\beta}(\lfloor{\sqrt{j+1}}\rfloor+3)^{-1/2}<\infty;
\]
hence, by~\cite[Theorem A.1]{HKMW} the space $(V,d_{\sigma_1})$ is not metrically complete. Let $D(x)$ be as in~(\ref{E:dist-boundary}) corresponding to $d_{\sigma_1}$. If a vertex $x$ is in the $n$-th row, using~(\ref{E:sigma-1}) and
\[
\lfloor{\sqrt{j+1}}\rfloor+3\leq 3\sqrt{j+1},\qquad\textrm{for all }j\in\mathbb{Z}_{+},
\]
we have
\[
D(x)\geq \frac{1}{\sqrt{3}}\sum_{k=n}^{\infty}(j+1)^{-\beta-\alpha/2-1/4}\geq \frac{(n+1)^{-\beta-\alpha/2+3/4}}{\sqrt{3}(\beta+\alpha/2-3/4)},
\]
which leads to
\begin{equation}\label{E:D-x-estimate}
\frac{1}{2{D(x)}^2}\leq \frac{3(4\beta+2\alpha-3)^2(n+1)^{2\beta+\alpha-3/2}}{32},
\end{equation}
for all vertices $x$ in the $n$-th row, where $n\geq 1$. Define $W(x)=-n^{2\beta+\alpha-3/2}$ for all vertices $x$ in the $n$-th row, where $n\geq 1$. Using~(\ref{E:minorant-eff}) and $W_{e}(x_{1,1})=b_1/2$, together with~(\ref{E:D-x-estimate}) and the assumption $0<\beta<3/4$, it follows that there exists a constant $C>0$ (depending on $\alpha$ and $\beta$) such that~(\ref{E:potential-minorant-effective}) is satisfied. Thus, by Theorem~\ref{T:main-2} the operator $\delswa+W$ is essentially self-adjoint on $\vcomp$. Clearly, Theorem~\ref{T:main-2} is also applicable in the case $W(x)=0$ for all $x\in V$, that is, the operator $\delswa$ is essentially self-adjoint on $\vcomp$. A calculation shows that $\mu$ and $b$ in this example do not satisfy~\cite[Assumption A]{Milatovic-11}; hence, we cannot use~\cite[Theorem 1.2]{Milatovic-11}.}

\emph{We will now show that under more restrictive assumption $1/2<\beta<3/4$, we cannot apply~\cite[Proposition 2.2]{Golenia-11} to this example with $W(x)\equiv 0$. To see this, using~(\ref{E:Deg-def})  and the fact that among the $\lfloor\sqrt{j}\rfloor +3$ edges departing from  the
vertex $x_{j,1}$, there are  $\lfloor\sqrt{j}\rfloor +1$
edges with weight $b_j$ and 2 edges with weight $b_{j-1}$, we first note that
\[
\operatorname{Deg}(x_{1,1})=2,\quad \operatorname{Deg}(x_{j,1})=j^{2\beta}((\lfloor{\sqrt{j}}\rfloor+1)j^{\alpha}+2(j-1)^{\alpha}),\quad\textrm{for all } j\geq 2.
\]
Let $\lambda\in\mathbb{R}$ be such that~(\ref{E:g-condition-1}) is satisfied, with $W(x)\equiv 0$. Let $\delta>0$ and let $a_n$ be as in~(\ref{E:g-condition-2}) corresponding to the path $\gamma=(x_{1,1};\,x_{2,1};\, x_{3,1};\dots)$, the potential $W\equiv 0$, $\delta>0$, and $\lambda$. Then $a_1=1$,
\[
(a_2)^2=\left(\frac{\delta}{2}+\left|1+\frac{\lambda}{2}\right|\right)^2=\frac{(\delta+|2+\lambda|)^2}{4},
\]
and
\begin{eqnarray*}
(a_n)^2&=&\frac{(\delta+|2+\lambda|)^2}{4} \prod_{j=2}^{n-1} \left(\frac{\delta+|j^{\alpha+2\beta}(\lfloor{\sqrt{j}}\rfloor+1)+2j^{2\beta}(j-1)^{\alpha}+\lambda|}{j^{\alpha+2\beta}(\lfloor{\sqrt{j}}\rfloor+1)
+2j^{2\beta}(j-1)^{\alpha}}\right)^2,\qquad n\geq 3.
\end{eqnarray*}
Therefore,
\[
\sum_{n=1}^{\infty}(a_{n})^2\mu(x_{n,1})= 1+\frac{(\delta+|2+\lambda|)^2}{4(2)^{2\beta}}+\sum_{n=3}^{\infty}\frac{(a_n)^2}{n^{2\beta}}.
\]
Using Raabe's test, it can be checked that the series on the right hand side of this equality converges. (Here, we used the more restrictive assumption $1/2<\beta<3/4$.) Hence, looking at~(\ref{E:g-condition-2}), we see that~\cite[Proposition 2.2]{Golenia-11} cannot be used in this example.}
\end{ex}

\begin{ex}\label{E:ex-3} \emph{Consider the graph whose vertices are arranged in a ``triangular" pattern so that $x_{1,1}$ is in the first row, $x_{2,1}$ and $x_{2,2}$ are in the second row,  $x_{3,1}$, $x_{3,2}$, and $x_{3,3}$ are in the third row, and so on.  The vertex $x_{1,1}$ is connected to $x_{2,1}$ and $x_{2,2}$. The vertex $x_{2,i}$, where $i=1,2$, is connected to every vertex $x_{3,j}$, where $j=1,2,3$. The pattern continues so that each of $k$ vertices in the $k$-th row is connected to each of $k+1$ vertices in the $(k+1)$-st row. Note that for all $k\geq 1$ and $j\geq 1$ we have $\textrm{deg}(x_{k,j})=2k$, where $\textrm{deg}(x)$ is as in~(\ref{E:ex-intrinsic}). Let $\mu(x)=k^{1/2}$ for every vertex $x$ in the $k$-th row, and let $b(x,y)\equiv 1$ for all vertices $x\sim y$. Following~(\ref{E:ex-intrinsic}), for every vertex $x$ in the $k$-th row and every vertex $y$ in the $(k+1)$-st row, define
\[
\sigma(x,y):=\min\left\{\frac{k^{1/2}}{2k},\frac{(k+1)^{1/2}}{2(k+1)}\right\}^{1/2}=2^{-1/2}(k+1)^{-1/4}.
\]
For all vertices $x$ in the $k$-th row, define $W(x)=-2k^{1/2}$ and $q(x)=2k$. Clearly, the inequality~(\ref{E:W-minorant-1}) is satisfied.
With this choice of $q$, following~(\ref{E:def-sigma-q}), for every vertex $x$ in the $k$-th row and every vertex $y$ in the $(k+1)$-st row, define
\[
\sigma_{q}(x,y):=\min\{(2k)^{-1/2},(2(k+1))^{-1/2}\}\cdot\sigma(x,y)=2^{-1}(k+1)^{-3/4}.
\]
Since
\[
\sum_{j=1}^{\infty}2^{-1}(j+1)^{-3/4}=\infty,
\]
by~\cite[Theorem A.1]{HKMW} it follows that the space $(V,d_{\sigma_q})$ is metrically complete. Additionally, it is easily checked that~(\ref{E:q-1-2-lipschitz}) is satisfied with $K=1$. Therefore, by Theorem~\ref{T:main-3} the operator $\Delta_{b,\mu}+W$ is essentially self-adjoint on $\vcomp$. Furthermore, it is easy to see that for every $c\in\mathbb{R}$, there exists a function $u\in\vcomp$ such that the inequality
\begin{equation}\nonumber
((\Delta_{b,\mu}+W) u,u)\geq c\|u\|^2
\end{equation}
is not satisfied. Thus, the operator $\Delta_{b,\mu}+W$ is not semi-bounded from below, and we cannot use~\cite[Theorem 1.2]{Milatovic-11}.}

\emph{It turns out that~\cite[Proposition 2.2]{Golenia-11} is not applicable in this example. To see this, using~(\ref{E:Deg-def}) we first note that $\operatorname{Deg}(x_{k,j})=2k^{1/2}$, for all $k\geq 1$ and all $j\geq 1$. Let $\lambda\in\mathbb{R}$ be such that~(\ref{E:g-condition-1}) is satisfied, with $W$ as in this example. Let $a_n$ be as in~(\ref{E:g-condition-2}) corresponding to the path $\gamma=(x_{1,1};\,x_{2,1};\, x_{3,1};\dots)$, the potential $W(x_{k,1})=-2{k}^{1/2}$, $\delta>0$, and $\lambda$. Then $a_1=1$, and for $n\geq 2$ we have
\begin{equation}\nonumber
(a_n)^2=\prod_{k=1}^{n-1}\left(\frac{\delta}{2k^{1/2}}+\left|1+\frac{\lambda-2k^{1/2}}{2k^{1/2}}\right|\right)^2=
\prod_{k=1}^{n-1}\frac{(\delta+|\lambda|)^2}{4k}=\frac{(\delta+|\lambda|)^{2n-2}}{4^{n-1}(n-1)!}.
\end{equation}
Therefore,
\[
\sum_{n=1}^{\infty}(a_{n})^2\mu(x_{n,1})= 1+\sum_{n=2}^{\infty}\frac{\sqrt{n}\cdot(\delta+|\lambda|)^{2n-2}}{4^{n-1}(n-1)!}.
\]
Using ratio test, it can be checked that the series on the right hand side of this equality converges. Hence, looking at~(\ref{E:g-condition-2}), we see that~\cite[Proposition 2.2]{Golenia-11} cannot be used in this example.}
\end{ex}

\subsection*{Acknowledgment} The second author is grateful to Daniel Lenz for fruitful discussions.

\end{document}